\documentclass[conference]{IEEEtran}
\ifCLASSINFOpdf
\else
\fi

\usepackage{datetime}
\usepackage{times}
\usepackage{latexsym,subfigure,makeidx}
\usepackage{epsfig,graphicx}  %amsthm
\usepackage{graphicx, amssymb}
\usepackage{amsfonts, enumerate}
\usepackage{url}
\usepackage{tabls}
\usepackage{latexsym, makeidx}
\usepackage{algorithm}
\usepackage{algpseudocode}
\usepackage{lscape}
\usepackage[sort,space]{cite}
\usepackage{pifont}
\usepackage{tabularx}
\usepackage{rotating}
\usepackage{hhline}
\usepackage{textcomp,booktabs}
\usepackage{colortbl}
\usepackage{array}
\usepackage{amsmath,bm}
\usepackage[]{caption2}
\usepackage{multirow}
\usepackage{verbatim}
\usepackage{amsthm}
\usepackage{framed}
\usepackage{geometry}

\newtheorem{lemma}{Lemma}

\ifCLASSINFOpdf
 \graphicspath{{../figure/}}
\else
\fi
\begin{document}
	\title{Joint Offloading and Resource Allocation  in Vehicular Edge Computing  and Networks}

	%\author{\IEEEauthorblockN{Yueyue Dai}

	\author{
		\IEEEauthorblockN{Yueyue Dai, Du Xu,
		}
		\IEEEauthorblockA{Key Laboratory of Optical Fiber Sensing and Communications,\\University of Electronic Science and Technology of China\\
		Email:$\{$daiyyue,xudu.uestc$\}$@gmail.com
		}
		%\IEEEauthorblockA{$^2$ School of Information Technology, Deakin University, Victoria, 3125, Australia
		
		%	\vspace{-0.2in}
%	}

   \and
	\IEEEauthorblockN{Sabita~Maharjan}
	\IEEEauthorblockA{Simula Research Laboratory\\
		Norway\\
		Email: sabita@simula.no}
	\and
		\IEEEauthorblockN{Yan Zhang}
		\IEEEauthorblockA{University of Oslo\\
			Norway\\
			Email:  yanzhang@ieee.org}
			}
	
	% after \
	
	%Telephone: (800) 555--1212, Fax: (888) 555--1212}

	% use for special pper notices

	\maketitle

	%\input{simulation.tex}
	\begin{abstract}
The emergence of computation intensive on-vehicle applications poses  a significant  challenge to provide the required computation capacity and maintain  high performance. Vehicular Edge Computing (VEC) is a new computing paradigm with a high potential to improve vehicular services by offloading computation-intensive tasks to the VEC servers. Nevertheless, as the computation resource of each VEC server is limited, offloading may not be efficient if all vehicles select the same VEC server to offload their tasks. To address this problem, in this paper, we propose  offloading with resource allocation.   We incorporate the communication and computation to derive the task processing delay. We formulate the problem as a system utility maximization problem, and then develop a low-complexity algorithm to jointly optimize  offloading decision and resource allocation.  Numerical results  demonstrate the superior performance of our Joint Optimization  of Selection and  Computation  (JOSC) algorithm compared to state of the art solutions.
	\end{abstract}

	% \fi
	\IEEEpeerreviewmaketitle

	\section{Introduction}
The advancements in  Internet of Things (IoT) and wireless technologies  have paved a way towards realizing new applications with advanced features. For instance, on-vehicle cameras and embedded sensors,  can play a crucial role towards efficient and safe transportation systems. However, the  resource-constrained vehicles can be strained by  computation-intensive applications,  resulting in bottlenecks and making it challenging for the vehicles to ensure  the required level of Quality of  Service (QoS) \cite{wang2017survey}. Mobile Edge Computing (MEC) can alleviate the need of heavy computation from the vehicles, yet enable such applications by providing  computation capabilities at the edge of  the radio access network and in close proximity to mobile users \cite{wang2017survey},\cite{mach2017mobile}. %to enhance the  vehicles to offload intensive tasks to nearby  servers located at the edge of radio access network, e.g. cell

% Unlike conventional MEC \cite{chen2016efficient}, \cite{zhang2016energy},\cite{you2017energy}, in vehicular network, multiple VEC servers are deployed on  nearby  road side units (RSUs)  to  offer vehicles a more closer and distributed end-user experience. Further,  awareness about location, \cite{Huang2017}.

% However, with high movement,  to make  more computation-intensive applications of vehicles could be accomplished timely,  computation offloading need to be elaborately  designed.

%In this paper, 

% 2) with high movement,  to make  more computation-intensive applications of vehicles could be accomplished timely,  computation offloading need to be elaborately  designed, 3) in vehicular communication network, dedicated short-range communication (DSRC) is one of the most reliable wireless technologies, providing two-way, short-range communication between vehicles and infrastructures

Computation  offloading  is a process where mobile users offload their  computation-heavy and latency-sensitive tasks to  base stations (BSs) for edge execution \cite{mach2017mobile}.   There has been considerable amount of work focusing on computation offloading under MEC where each user independently chooses whether to execute the task locally or to offload the task to edge servers, to minimize  energy consumption and/or computation latency \cite{you2017energy}, \cite{yu2016joint},\cite{wang2017joint2},\cite{mao2016dynamic}.   In \cite{you2017energy} and \cite{yu2016joint}, the offloading problem was studied with an objective to minimize the weighted  energy consumption under the constraint on computation latency.  The authors in\cite{wang2017joint2} proposed a multi-user MEC system by integrating  a multi-antenna Access Point (AP)  with an MEC server by jointly optimizing offloading and computing. In \cite{mao2016dynamic}, a delay-optimal computation offloading algorithm was proposed by   jointly considering the offloading decision, the CPU-cycle frequencies,  and the transmit power.  In general, these schemes consider there is only one MEC server located at BS and  all mobile users have to offload their tasks to  the only one MEC server to execute computation.  %Therefore, the above  state of the art schemes either infeasible or of limited applicability in vehicular  network with  multiple VEC servers.
With the limitation of the amount of MEC servers, some users' tasks may not be accomplished within the permissible latency threshold.  

 Vehicular Edge Computing (VEC) is a promising new paradigm that has received much attention lately, as it can extend the computation  capability to vehicular network \cite{zhang2017mobile}.  In conventional MEC network, all mobile users have to offload their tasks to the only one MEC server located at the BS. Different from MEC, in VEC,  lightweight but ubiquitous edge resources deployed at nearby Road Side Units (RSUs) can offer high QoS.  Further, localized processing is enabled to save backhaul bandwidth and awareness about location is also beneficial  to resource allocation.    
 %With the limitation of the amount of MEC servers and computation resource, some users' tasks may not be accomplished within the permissible latency threshold.  

% In conventional MEC network, all mobile users have to offload their tasks to  the only one MEC server located at the BS. With the limitation of the amount of MEC servers and computation resource,  some users' tasks may not be accomplished within the permissible latency threshold.   A vehicular network consists of  multiple VEC servers and each vehicle can select one of them as its target offloading server. Therefore, the above  state of the art schemes either infeasible or of limited applicability in VEC network. A few studies \cite{Zhang2017} \cite{Zhang2016} investigated the performance of computation offloading for vehicle networks. In \cite{Zhang2017}, the authors proposed the authors proposed a stackelberg game theoretic approach for offloading,  The authors in \cite{Zhang2016} designed a computation offloading strategy through a contract theoretic approach.   However, these schemes  may result in a severe overload where each vehicle greedily selects the VEC server with the highest utility. To overcome of this shortcoming, we  balance the load among VEC servers and, at the same time,  maximize system utility. In addition, as the mobility of vehicles also has a significant  impact on task completion latency, we also incorporate mobility into problem formulation.

  A few studies investigated computation offloading for vehicle networks. In \cite{Zhang2017}, the authors considered that a vehicular network consisted of  multiple VEC servers and each vehicle selected one of them as its target offloading server to maximize the utilities of both the vehicles and the computing servers. The authors in \cite{Zhang2016} proposed a similar  computation offloading strategy.   However, these schemes  may result in a severe overload as all vehicles greedily select the VEC server  with the highest utility. It further  leads to low offloading efficiency and prolongs task processing delay. To overcome  this shortcoming, we  propose an offloading scheme by balancing the load among VEC server, at the same time,  to maximize system utility. In addition, as the mobility of vehicles also has a significant  impact on task processing delay, we also incorporate the mobility aspect into problem formulation.
 %Moreover, with the increasement of offloaded tasks, some users' tasks may not be accomplished within the permissible latency threshold as MEC server has finite computation resource

%Therefore, mobile vehicles can potentially enhance their QoS by offloading tasks to the VEC server.

We propose integrating  offloading and resource allocation in order to carefully address the following critical challenges: 1) VEC server selection for offloading, 2) determining optimal computation resource to maximize system utility. To address these challenges, we consider a VEC network that multiple VEC servers equipped on RSUs locate along the road. Vehicles can select a VEC server to offload their tasks for satisfying stringent latency requirements. We model VEC server selection as a binary decision.  By jointly optimizing  selection offloading decision and computation resource allocation, we obtain the optimal strategy for maximizing system utility.  The key contributions of our work are as follows:

% To satisfy  stringent latency constraints, vehicles  select the most adaptive VEC server to offload their tasks to.

\begin{itemize}

\item We jointly model VEC server selection for  offloading, and resource allocation to execute tasks associated with vehicular applications.

\item We formulate the joint   offloading and resource allocation problem as a system utility maximization problem under the latency constraint.  %We adopt comprehensive task processing delay as the performance metric by jointly analyzing communication model   and computation model.

\item We propose a Joint Optimization of   Selection and  Computation  (JOSC) algorithm  to find the solution of the optimization problem in a distributed manner and with less overhead. %Extensive experiments demonstrate the effectiveness of proposed schemes.
\end{itemize}
The rest of this paper is organized as follows. In Section \ref{sm}, we introduce the system model. In Section \ref{pF}, we formulate the joint  offloading and resource allocation problem  and propose a low-complexity JOSC algorithm to solve this problem. We evaluate the performance of the JOSC algorithm and provide illustrative results in Section \ref{sr}. We conclude the paper in Section \ref{c}.

	\section{System Model }
\label{sm}

%\subsection{System model}
%We consider an unidirectional road, where one Macro Base Station (MBS) and $M$ Road Side Units (RSUs) locate long the road, as shown in Fig. \ref{overview}. Each RSU is equipped with a VEC server, whose computation resources are limited. The MBS is equipped with a MEC server. We denote the  id set of these RSUs and MBS as $\mathcal{M}=\{0,1,...,M\}$, where  the id of the MBS is $0$.
 %Due to the different radio power of each RSU and the variety of the wireless environment, the size of the wireless coverage areas of these RSUs may be different \cite{zhang2011home}.  We divided the road as $M$ segments, which length are $\{L_1,L_2,..,L_M\}$ respectively. %The computation capacity of the VEC sers is $\{F_1,...F_M\}$, respectively. 
We consider a unidirectional road, where $M$ RSUs are located along the road, as shown in Fig. \ref{overview}.  The road is correspondingly divided into $M$ segments, with length $\{L_1,L_2,..,L_M\}$ respectively. Each RSU is equipped with a VEC server, whose computation resources are limited.  We denote the  id set of these RSUs as $ \mathcal{M}=\{1,...,M\}$.  %and  the id of the MBS as $1$.
  %The computation resource of the VEC sers is $\{F_1,...F_M\}$, respectively. 

There are  $N$ vehicles arriving at the starting point of the road.  The vehicles are running at speed $v$. Each vehicle has a computation task to be completed with a stringent delay constraint. The tasks include interactive gaming,  real-time financial trading, virtual reality and etc.  Each computation task  can be described  as  $D_i\triangleq(d_i,c_i,T_{i}^{max}), i \in \mathcal{N}=\{1,2,...,N\}$, where  $d_i$ denotes the size of computation input data (e.g. the program codes and input parameters), $c_i$ is the  required computation resource for computing  task $D_i$, and $T_{i}^{max}$ denotes the  maximum latency  allowed to accomplish the task.  %Each task can be divided into two parts and  parallel executed at the vehicles and VEC servers. Specifically, let $\lambda_{ij} (0\leq \lambda_{ij} \leq1) $ be the offloading ratio variable which represents the ratio of  the offloaded task to the total task.  Vehicle $i$ offloads  $\lambda_{ij} d_i $ to the VEC server on RSU $j$ and computes the rest $(1-\lambda_{ij})d_i$ locally \cite{you2017energy}.  %We assume  that  each vehicle could only select one VEC server as its target offloading server. %$\lambda_{ij} d_i$  can be offloaded to only one  VEC server. %edge server that may be the MEC server on MBS or VEC server on RSU. 

Each task can either  be offloaded  to a selected VEC server to process, or be executed locally at the vehicle. Denote the selection decision variables by  $x_{ij}$ and $x_{i0}$ indicating whether the task of  vehicle $i$ is processed on the selected VEC server on RSU $j$ (i.e. $x_{ij}=1$) or locally (i.e. $x_{i0} = 1$).  The selection decision variables should satisfy the constraint $\sum_{j=0}^{M}x_{ij} =1$ which indicates  that only one of $x_{ij}$ could  be $1$.
\begin{figure}
	\centering
	\includegraphics[width =3.6 in]{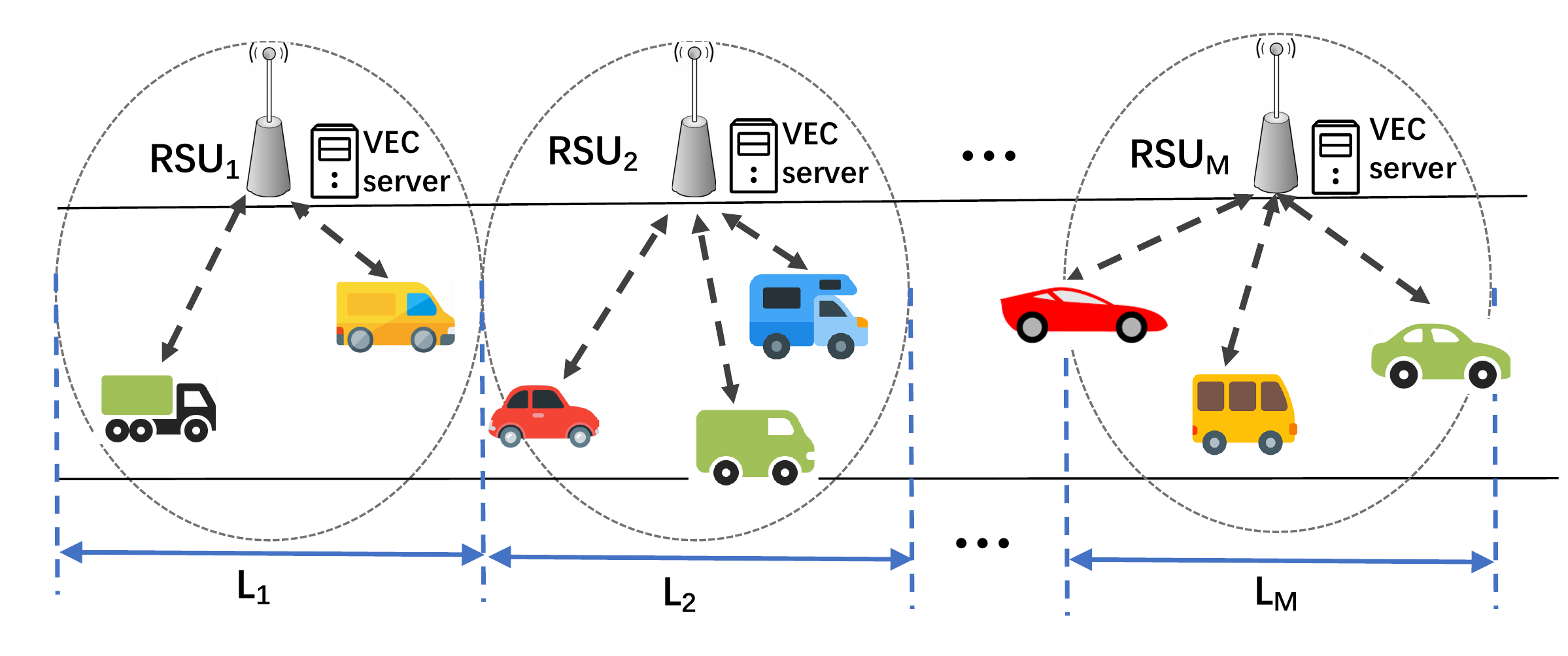}
	\caption{The VEC offloading in  a  vehicular network.}
	\label{overview}
\end{figure}

\subsection{Offloading to  VEC Server}
If vehicle $i$ chooses to offload  task $D_i$  to a selected VEC server to  process, the  task processing  delay  can be divided into four parts.  The first one is the time taken for vehicle $i$ from the starting point to  the coverage of  RSU $j$ (i.e. $\sum_{k=1}^{j}L_k)/v$). The second part is the  communication time that the task is transmitted from vehicle $i$ to RSU $j$. The third component  is the computation time.  The fourth part is  the communication time that the  result of the task is transmitted from RSU $j$ to vehicle $i$. However, since the size of the result is often much smaller than the input data size,  we do not consider the energy consumption and latency of this part \cite{chen2016efficient}, \cite{wang2017joint}, \cite{zhang2016energy}.

\emph{{1) Communication Time:}}

We assume the wireless  communication between  vehicles and  RSUs  is based on the Orthogonal Frequency Division Multiple Access (OFDMA).  The RSUs in this system are also allocated orthogonal spectrum such that the inter-cell interference among RSUs can be ignored.

Let $h_{ij}$ denote the channel gain and $p_{ij}$ the transmission power for vehicle $i$. Then the achievable rate, denoted by $r_{ij}$, is given as:
\vspace{-0.1in}
\begin{equation}
r_{ij}=B\log(1+\dfrac{p_{ij}h_{ij}}{N_0})
\end{equation}
where $N_0$ is the noise power  and $B$ is the system bandwidth.

The communication time for offloading $d_i$ from vehicle $i$ to RSU $j$ can be written as
\vspace{-0.1in}
\begin{equation}
\label{comm}
T_{ij}^{com} = \frac{d_i}{B\log(1+\dfrac{p_{ij}h_{ij}}{N_0})}
\end{equation}

\emph{{2) Computation  Time:}}

Since  multiple vehicles may choose  the same VEC server as their offloading target and the computation resource of each VEC server is limited, we need to allocate this computation resource to satisfy all vehicles' latency constraints. Denote $F_j$ as the total computation resource of VEC server on RSU $j$ and  $f_{ij}$ is the amount of  computation resource  that the VEC server assigns to vehicle $i$. We have $\sum_{i=1}^{N}x_{ij}f_{ij}\leq F_j$. The computation execution time $T_{ij}^{vec}$ will be
\vspace{-0.05in}
\begin{equation}
\label{RSU}
T_{ij}^{vec}=\dfrac{c_i}{f_{ij}}
\end{equation}

Then, the task processing delay for offloading $D_i$ from vehicle $i$ to RSU $j$ can be expressed as
\vspace{-0.05in}
\begin{equation}
T_{ij} = \sum_{k=1}^{j}\dfrac{L_k}{v}+T_{ij}^{com} +T_{ij}^{vec}
\end{equation}

\subsection{Local Processing }
If vehicle $i$ executes  its computation task $D_i$ locally, the task processing delay is determined by its own computation resource. Let $f_i$ denote the computational resource of vehicle $i$, which varies for different users and can be obtained through offline measurement \cite{miettinen2010energy}. The local computation execution delay $T_{i0}$  can now be expressed as
\vspace{-0.05in}
\begin{equation}
\label{local}
	 T_{i0}=\dfrac{c_i}{f_i}
\end{equation}

%\begin{figure}
%	\centering
%	\includegraphics[width =1.5 in]{vehicle_offload.pdf}
%	\caption{The VEC-MEC offloading in  a  vehicluar network.}
%	\label{overview}
%\end{figure}

\vspace{-0.05in}

\subsection{System Utility Function}
Different from the task offloading in previous works which aim to optimize  energy consumption \cite{you2017energy}, \cite{yu2016joint},  \cite{wang2017joint2}, the task processing time is a critical metric for vehicles. We therefore design a QoS based utility function, i.e., based on the task processing time.

Since each task can  either be offloaded to a selected VEC server to process, or be executed locally, the task processing time can be expressed as 
\vspace{-0.05in}
\begin{equation}
\label{delay}
\begin{split}
T_i& =\sum_{j=1}^{M}x_{ij}T_{ij}+x_{i0}T_{i0}= \sum_{j=0}^{M}x_{ij}T_{ij}
\end{split}
\end{equation}
%Note that, if there is no adpative VEC server which could provide

Due to the fact that  moving vehicles may have a higher satisfaction with a smaller $T_i$, the utility function  should monotonically decrease with $T_i$. Moreover,   because of the limitation of computation resource, offloading can be less efficient and result in overload, if all vehicles select the same VEC server to offload their task to.  The utility function also should balance the load among VEC servers. According to \cite{kelly1997charging}, the logarithmic utility is known as proportional fairness, which is able to achieve load balancing. Therefore,  we define the utility function as 
\vspace{-0.1in}
\begin{equation}
\label{u_i}
\begin{split}
U_i(x_{ij},f_{ij}) &=\alpha \log(1+\beta-T_i)\\
         % & =\alpha \log [1+\beta-\frac{(1-\lambda_{ij}) 
        %  & =  \alpha \log [1+\beta-\zeta-\lambda_{ij}d_ic_i\left({1/f_{ij}}-{1/f_i}\right)-T_{ij}]
\end{split}
\end{equation}
where $\alpha$ is a satisfaction parameter, $\beta$ is used to normalize  the satisfaction to be nonnegative. The higher the $\alpha$, the more gain of satisfaction. 
The system utility function can be written as $U = \sum_{i=1}^{N}U_i$.
%\vspace{-0.1in}
%\begin{equation}
%\label{su}
%U = \sum_{i=1}^{N}U_i
%\end{equation}
	 \section{Problem Formulation and Solution}
 
\label{pF}
In this section, we formulate the joint  offloading and resource allocation scheme as an optimization problem. The objective is to maximize the system utility. %The partical offloading supports parallelism. 
Define $\mathbf{x}=\{x_{ij}\}$ as the vector of VEC server selection decision and $\mathbf{f}=\{f_{ij}\}$ as the computation resource vector, respectively. We  formulate the optimization  problem as follows:
\begin{subequations}
	\label{S1}
	\begin{align}
	\centering
		\max \displaystyle 
	~~~&U(\mathbf{x},\mathbf{f}) 	= {\sum_{i=1}^{N}{\alpha \log(1+\beta-T_i)}}\\    
	               % &~~~~~~~ = \sum_{i=1}^{N}\alpha \log(1+\beta-\sum_{j=1}^{M}\max\{x_{ij} T_{ij}^{loc},\\& ~~~~~~~~~~~~x_{ij}(\sum_{k=1}^{j-1}\dfrac{L_k}{v}+T_{ij}^{trans}+T_{ij}^{vec})\})\\    
			s.t.~~ 
	          &~T_i% = \sum_{j=1}^{M}\max\{ x_{ij}T_{ij}^{loc}, x_{ij}(\sum_{k=1}^{j-1}\dfrac{L_k}{v}+T_{ij}^{trans}\\&~~~~~~~~~~~~~~+T_{ij}^{vec})\}
			\leqslant T_i^{max},~~\forall~ i \in \mathcal{N} \label{c1}\\ 
			&~ \sum_{j=0}^{M}x_{ij}=1, ~~~~~~~~~\forall ~i  \in \mathcal{N} \label{c2}\\    	
			&~\sum_{i=1}^{N}x_{ij}f_{ij}\leqslant F_j,~\forall ~j\in \mathcal{M}\label{c3}\\
		%	&~\sum_{j=1}^{M}x_{ij}\lambda_{ij}\leqslant 1~\forall ~i\in \mathcal{N} \label{c4}\\
			&~0\leqslant f_{ij} \leqslant F_j~~~~~~~\forall ~i  \in \mathcal{N} ,~j\in \mathcal{M}\label{c5}\\
		%	&~0\leqslant \lambda_{ij}\leqslant 1~~~~~~~\forall ~i  \in \mathcal{N} ,~j\in \mathcal{M}\label{c6}\\
			&~x_{ij}\in\{0,1\}, ~~~~\forall ~i  \in \mathcal{N} ,~j\in \mathcal{M}  \cup\{0\}  \label{c7} 
	\end{align}
\end{subequations}
The first constraint (\ref{c1}) guarantees that the task processing time  cannot exceed the maximum allowed latency  $T_i^{max}$. Constraints (\ref{c2})  and (\ref{c7}) state each vehicle offloads its task to one and only one VEC server.  Constraints (\ref{c3}) and (\ref{c5}) ensure that the sum of the computation resource assigned to all tasks, which choose the  VEC server on RSU $j$,  does not exceed the total computation capacity of this VEC server.  %Constraints (\ref{c4}) presents the total offloaded task of vehicle $i$ cannot exceed to $1$. T
The key challenge in solving this problem is the integer constraint $x_{ij}\in \{0,1\}$, which  makes  $(\ref{S1})$ a mixed-integer non-linear programming problem and this is in general  non-convex and NP-hard  \cite{boyd2004convex}.
 
 In order to solve  (\ref{S1}), we first transform  it  into an equivalent form as shown in Lemma $\ref{le1}$. 
 
\begin{lemma}
\label{le1}
	The optimization problem (\ref{S1}) can be  transformed into the following equivalent problem:
	\begin{subequations}
	\label{S2}
		\begin{align}
		\centering
		\max \displaystyle 
				~~~&U(\mathbf{x},\mathbf{f}) 	= {\sum_{i=1}^{N}\sum_{j=0}^{M}{\alpha \log(1+\beta-x_{ij}T_{ij})}}\\   
				s.t.~~ 
				          &~\sum_{j=1}^{M}x_{ij}(\Lambda_{ij}+\dfrac{c_i}{f_{ij}})\leqslant T_{i}',~~\forall~ i \in \mathcal{N} \label{c8}\\ 
				         &(\ref{c2}), (\ref{c3}), (\ref{c5}), (\ref{c7})\notag
		\end{align}
	\end{subequations}
	where $\Lambda_{ij} = \sum_{k=1}^{j}\dfrac{L_k}{v}+T_{ij}^{com}-\dfrac{c_i}{f_i}$ and  $T_i' = T_{i}^{max}-T_{i0}$.
\end{lemma}
\begin{proof}
See Appendix \ref{appA}.
\end{proof}

It is still challenging to  solve (\ref{S2}) which has  non-linear objective function  and integer constraint (\ref{c7}). Therefore, we  decouple  selection decision and computation resource allocation  into two  subproblems to develop a low-complexity algorithm.  That is, we determine $\mathbf{x}$ under  given  $\mathbf{f}$  and then $\mathbf{f}$    under  obtained $\mathbf{x}$, and repeat this process until convergence.

\subsection{Selective Decision}

The selection decision problem for a given $\mathbf{f}$ from (\ref{S2}) takes the form

\begin{equation}
	\begin{split}
	\max	~~~&U(\mathbf{x})  = {\sum_{i=1}^{N}\sum_{j=0}^{M}{\alpha \log(1+\beta-x_{ij}T_{ij})}}\\   
   s.t. ~~&(\ref{c2}),(\ref{c3}),(\ref{c7}),(\ref{c8})
	\end{split}
	 \label{ua}
\end{equation}

In (\ref{ua}), all the indicator variables $x_{ij}$s are binary, while $U(\mathbf{x}) $ is a non-linear function with respect to $x_{ij}$. Thus it is a Mixed Integer Non-Linear Programming problem (MINLP), which is usually NP-hard. 

To solve this problem with low complexity, we propose an approximation algorithm.  First, we construct a subset $\mathcal{B}_i$ with a latency threshold. Then, we relax the original MINLP problem as a non-linear programming problem (RNLP) and solve it using standard convex method. Finally, we use  rounding to obtain a feasible solution. %Next, we describe the details.

To ensure each task of vehicles can be accomplished respecting its latency deadline, given $\mathbf{f}$, we construct  an available subset $\mathcal{B}_i $ 

\begin{equation}
\label{sub}
\mathcal{B}_i =\mathcal{M}\cap\{j| \dfrac{T_{ij}}{ T_{ij}^*}\leqslant \rho\}
\end{equation}
where $T_{ij}^*= \min_{j\in \mathcal{M}}T_{ij}$ and $\rho$ is a threshold. Usually only a limited number of RSU will satisfy the above constraint  (\ref{sub}). For the RSU that cannot meet the threshold, let $x_{ij}=0,j\notin \mathcal{B}_i$. Thus, the complexity of (\ref{ua}) will be greatly reduced. %a limited number of RSUs will be taken into consideration for a vehicle. %We can see not only task's latency requirements will be satisfied, but also the computational complexity will be greatly reduced.
%\begin{lemma}
%The optimal solution of MINLP problem (\ref{ua}) is uppo
%\end{lemma}

Once    $\mathcal{B}_i$ is determined, we relax the binary variable $x_{ij}$ into a real value in $[0,1]$. Then the MINLP problem (\ref{ua}) is relaxed into an RNLP as follows:

\begin{subequations}
  \label{ua2}
	\begin{align}
	\centering
	\max	~~~&U(\mathbf{x}) 	= {\sum_{i\in\mathcal{N}}\sum_{j\in \mathcal{B}_i}{\alpha \log(1+\beta-x_{ij}T_{ij})}}\\   
	   s.t. ~~
	    &~\sum_{j\in \mathcal{B}_i}x_{ij}(\Lambda_{ij}+\dfrac{c_i}{f_{ij}})\leqslant T_{i}',~~\forall~ i \in \mathcal{N} \\     
	    %&\sum_{i \in \mathcal{N}}x_{ij}f_{ij}\leqslant F_j,~~~~\forall j\in \mathcal{B}_i\\%\forall j\in \mathcal{B}_i, ~~~\forall i \in \mathcal{N}\\
	   &~x_{i0}+ \sum_{j\in \mathcal{B}_i}x_{ij}=1, ~~~~~~~~~\forall ~i  \in \mathcal{N} \label{xc}\\	  
	             &x_{ij}\geqslant 0,~~~~~~\forall j\in \mathcal{B}_i\cup\{0\}, ~~~\forall i \in \mathcal{N} \label{c12}\\
	             &(\ref{c3})\notag
	\end{align}
\end{subequations}
where $  x_{ij} = 0, j\notin \mathcal{B}_i$.

Since the sum of $x_{ij}$ is already upper bounded by $1$ in   (\ref{xc}), we remove the upper bound $1$ of $x_{ij}$ and obtain the constraint (\ref{c12}). The objective function   $U(\mathbf{x})$ is  concave. Combining with the linear convex constraints, problem (\ref{ua2}) is a convex optimization problem and we can solve (\ref{ua}) to obtain an optimal fractional solution, which  is an upper bound of the original MINLP problem, because it is obtained by expanding the solution space.

We  denote the fractional solution of (\ref{ua2})  as $\mathbf{x}' = \{x_{ij}'|x_{ij}'\in [0,1]\}$.  The solution of RNLP is usually an infeasible solution to the original MINLP problem as it is fractional. Therefore, we adopt  a rounding method from \cite{shmoys1993approximation} to obtain a feasible solution for the MINLP problem. In this rounding method, a bipartite graph is constructed according to the RNLP solution, which is constructed as an undirected bipartite graph.% Next, we give the details about the rounding.

The rounding technique  consists of the following two steps: 1) construct a weighted bipartite graph to establish the relationship between vehicles and RSUs, 2) find a maximum matching to obtain the integer solution based on the bipartite graph.

In  step 1,  we  construct the weighted bipartite graph $\mathcal{G}(\mathcal{U},\mathcal{V},\mathcal{E})$  to establish the relationship  between vehicles and RSUs.  The set $\mathcal{U}$ represents the vehicles in the network.  The set $\mathcal{V} = \{v_{js}:j \in\mathcal{B}_i, s= 1,...,J_j\}$, where $J_j = \lceil\sum_{i =1}^{N}x_{ij}'\rceil$  implies  VEC server on RSU  $j$ serves $J_j$ vehicles. The nodes $\{v_{js:s = 1,..,J_j}\}$ correspond to RSU $j$. 
 
 The most important procedure for constructing $\mathcal{G}$ is  to set the edges and the edge weight between $\mathcal{U}$ and $\mathcal{V}$. The edges in $\mathcal{G}$ are constructed according to the following method. 
 
\begin{framed}
  Construct the edges of  bipartite graph $\mathcal{G}$: % constructed step:
  
 % 1. For each RSU, sort vehicles in the order of nonincreasing time $t_{ij}$.

   ~~ If $J_j\leqslant 1$, there is only one node $v_{j1}$ corresponding to RSU $j$. For each $x_{ij}'>0$, add edge $(u_i,v_{j1})$ and set the weight of this edge as $ x_{ij}'$. 
   
    ~~Otherwise,
    
   ~~~1) Find the minimum index $i_s$ such that  $\sum_{i=1}^{i_s}x_{ij}'\geqslant s$.
   
~~~ 2) For $i = i_{s-1}+1,..,i_{s}-1$, and $x_{ij}'>0$, add  edge $(u_i,v_{js})$ with weight $x_{ij}'$.
  
  ~~~~3) For $i = i_s$, add edge $(u_i,v_{js})$ with weight $1-\sum_{i=1}^{i_s-1}x_{ij}'$. This ensures that the total weight of edges connecting $v_{js}$ is at most 1.
   
  ~~~ 4) If $\sum_{i=1}^{j_s} x_{ij}'>s$, add edge $(u_i,v_{j(s+1)})$ with weight $\sum_{i=1}^{i_s}x_{ij}'-s$.  
  \vspace{-0.1in}
   \end{framed}

Based on the above steps, we construct a weighted bipartite graph $\mathcal{G}(\mathcal{U},\mathcal{V},\mathcal{E})$.   In step 2, we use the Hungarian algorithm \cite{kuhn1955hungarian} to find a maximum profit matching whose total profit is the maximum among all matchings. Note that, the \emph{profit}  of each edge is defined to be $\alpha\log(1+\beta-T_{ij})$. According to the matching, we obtain the integer  selection  decision. Specifically,  if the edge $(u_i,v_{js})$ is in the matching,  set $x_{ij}= 1$; otherwise, $x_{ij}= 0$.  This maximum matching indicates a feasible solution for the MINLP problem. According to \cite{shmoys1993approximation}, the solution produced by the rounding approximation algorithm is at most $(1+\rho)$ times greater than the optimal solution.

%Therefore, we adopt  a rounding method from \cite{shmoys1993approximation} to obtain a feasible solution for the MINLP problem. We recover the feasible decision $x_{ij}$ according to the following equation:
%\begin{equation}
%x_{ij} =\text{round}(x_{ij}'), ~~~\forall~j\in\mathcal{M}\cup\{0\}
%\end{equation}
%According to Theorem \ref{te1}, the solution produced by the rounding method is at most $(1+\rho)$ times greater than the  RNLP optimal solution.

\subsection{Optimization of  Computation Resource}
%\emph{{1)  Optimization of  Computation Resource:}} 

%Substituting (\ref{T_S}) into problem (\ref{joint}), 
The computation resource allocation problem for a given $\mathbf{x}$ is
%\vspace{-0.1in}
\begin{small}
\begin{equation}
  \label{computation}
	\begin{split}
	\centering
	\max	~&U({\mathbf{f}})	=\sum_{i=1}^{N}\sum_{j=1}^{M}\alpha \log(1+\beta-x_{ij}(\dfrac{c_i}{f_{ij}}
+\Lambda_{ij}))\\  
	   &    s.t.    % ~~~~&\sum_{j=1}^{M}x_{ij}( \lambda_{ij}(\dfrac{c_i}{f_{ij}}+\zeta_{ij})+\Lambda_{ij})  \leqslant T_i^{max},~~\forall~ i \in \mathcal{N} \label{27b}\\ 
	       % ~~~~&f_{ij}\geqslant 0\label{27c}\\
	         ~~~~(\ref{c3}),(\ref{c5}),(\ref{c8})
	\end{split}
\end{equation}
\end{small}
where $\sum_{i=1}^{N}\alpha\log(1+\beta-x_{i0}T_{i0})$ is omitted from (\ref{computation}), as it is a constant.
\begin{lemma}
\label{le3}
Problem (\ref{computation}) is a convex optimization problem.
\end{lemma}
\begin{proof}
See Appendix \ref{appB}.
\end{proof}

Since  (\ref{computation}) is a convex problem, we use  Lagrangian method to solve this problem.
The Lagrangian function is 
\begin{small}
\begin{equation}
    	\begin{split}
    	\mathcal{L}(\mathbf{f},\theta,\psi,\varrho,\omega) =& \sum_{i=1}^{N}\sum_{j=1}^{M}\alpha \log(1+\beta-x_{ij}(\dfrac{c_i}{f_{ij}}+\Lambda_{ij}))-\\&\sum_{i=1}^{N}\theta_i(\sum_{j=1}^{M}x_{ij}(\Lambda_{ij}+\dfrac{c_i}{f_{ij}})- T_i')-\\&\sum_{j=1}^{M}\psi_j(\sum_{i=1}^{N}x_{ij}f_{ij}-F_j)-\sum_{i=1}^{N}\sum_{j=1}^{M}\varrho_{ij}(f_{ij}-F_j)\\&+\sum_{i=1}^{M}\sum_{j=1}^{M}\omega_{ij}f_{ij}
    	\end{split}
    \end{equation}
\end{small}
where $\theta$, $\psi$, $\varrho$, and $\omega$ are the Lagrangian multipliers.  The Lagrange dual function is then given by:
\begin{equation}
\mathcal{D}(\theta,\psi,\varrho,\omega) = \max~ \mathcal{L}(\mathbf{f},\theta,\psi,\varrho,\omega)
\end{equation}
and the dual problem of (\ref{computation}) is
\begin{equation}
\label{dual}
\begin{split}
\min ~~~~&\mathcal{D}(\theta,\psi,\varrho,\omega)\\
s.t. ~~~&\theta\succ0, ~\psi,~\varrho,~\omega\succeq 0
\end{split}
\end{equation}

Since  (\ref{computation}) is convex, there exists a strictly feasible point where Slater's condition holds, leading to strong duality \cite{boyd2004convex}. This allows us to solve the primal problem (\ref{computation}) via the dual problem (\ref{dual}). The dual problem (\ref{dual}) can be solved using the gradient method. As the Lagrange function is differentiable, the gradients of  the Lagrange multipliers can be obtained as
\begin{equation}
\begin{split}
&\dfrac{\partial \mathcal{L}}{\partial \theta_i}   = -(\sum_{j=1}^{M}x_{{{\it ij}}} (\Lambda_{ij}+ {\dfrac {c_{{i}}}{f_{ij}}})-T_i')\\
&\dfrac{\partial \mathcal{L}}{\partial \psi_j}   =-(\sum_{i=1}^{N}x_{{{\it ij}}}f_{ij}-F_{{j}})\\
&\dfrac{\partial \mathcal{L}}{\partial \varrho_{ij}}   =-(f_{ij}-F_{{j}})\\
&\dfrac{\partial \mathcal{L}}{\partial \omega_{ij}}   =f_{ij}
\end{split}
\end{equation}

By applying the gradient method, the Lagrange multipliers are calculated iteratively as follows:
\begin{equation}
\label{iterative}
\begin{split}
&\theta_i(t+1) =\left[ \theta_i(t)+\kappa_1\dfrac{\partial \mathcal{L}}{\partial \theta_i} \right]^+\\
&\psi_j(t+1) =\left[ \psi_j(t)+\kappa_2\dfrac{\partial \mathcal{L}}{\partial \psi_j} \right]^+\\
&\varrho_{ij}(t+1) =\left[\varrho_{ij}(t)+\kappa_3\dfrac{\partial \mathcal{L}}{\partial\varrho_{ij}} \right]^+\\
&\omega_{ij}(t+1) =\left[\omega_{ij}(t)+\kappa_4\dfrac{\partial \mathcal{L}}{\partial\omega_{ij}} \right]^+\\
\end{split}
\end{equation}
where $\kappa_1,\kappa_2,\kappa_3,\kappa_4>0$ are the gradient steps, $t$ represents the gradient number, and $[\cdot]^+$ denotes $\max(0,\cdot)$.  %The remaining Lagrange multipliers $\psi_j$ and $\omega_j$ are obtained iteratively using similar equations.

Taking   the first-order derivative of  $\mathcal{L}$ with respect to $f_{ij}$ and setting the result to zero, we obtain,
\begin{small}
\begin{equation}
\label{f_op}
\begin{split}
\dfrac{\partial \mathcal{L}}{\partial f_{ij}}   = &{\dfrac {\alpha\,x_{{{\it ij}}}c_{{i}}}{f_{ij}^{2}\ln 
 \left( 2 \right)  \left( 1+\beta-x_{ij}
 \left( {\dfrac {c_{{i}}}{f_{ij}}}+\Lambda _{{{\it ij}}} \right)  \right) }}\\&+{\dfrac {\theta_i\,x_{{{\it ij}}}
c_{{i}}}{f_{ij}^{2}}}-\psi_{{j}}x_{{{\it ij}}}-\varrho_{ij}+\omega_{ij}
=0
\end{split}
\end{equation}
\end{small}

\subsection{ Joint Algorithm for Selection Decision and Computation Resource}
Based on above analysis, we present our joint algorithm for  selection decision and computation resource (JOSC), which is summarized in Algorithm \ref{outer}. According to  JOSC,  (\ref{S2}) can be solved in a semi-distributed manner.  Specifically,  based on  given  $\mathbf{f}$,  each vehicle obtains its  selection decision by solving relaxation problem (\ref{ua2}) and using rounding method. Then  under the obtained selection decision and  by exchanging computation resources among neighboring vehicles,  each vehicle uses Lagrangian method to obtain $\mathbf{f}$  and repeat this process until convergence. 

At  each iteration of Algorithm \ref{outer}, the computational complexity of solving convex problem (\ref{ua2}) is only polynomial in the number of variables and constraints.  The complexity required to solve (\ref{ua2}) is thus $\mathcal{O}((1+a+b)a^2\sqrt{b+1})$, where $a = N*M'$ is the number of decision variables, $ b  =3N+M+M'$ is the number of  linear constraints and  $M' =|\mathcal{B}_i| $  ( $|\cdot|$ denotes the cardinality of a set).  The complexity of rounding is  polynomial in the number of nodes and edges,  that is $\mathcal{O(|\mathcal{V}||\mathcal{E}|)}$. For $t$ iterations, the complexity of the inner loop of Algorithm \ref{outer} is $\mathcal{O}(t)$. Thus,  the total complexity of Algorithm \ref{outer}, for $k$ iterations, is $ \mathcal{O}(k(t+(1+a+b)a^2\sqrt{b+1}+|\mathcal{V}||\mathcal{E}|)$.

\begin{algorithm}[!h]

	\caption{ {J}oint  Optimization for {S}election and {C}omputation algorithm  (JOSC)}
		\label{outer}
	\begin{algorithmic}[1]
		\State  \text{Initialization}:
		   \State   ~~~Set  selection decision  $\mathbf{x}^{(0)}$ and computation resource  $\mathbf{f}^{(0)}$ to arbitrary values;
		       \State   ~~  Set the number of iteration as $k = 0$;
		        %%%%%    \end{itemize}
		            % Initialize the Lagrangian multipliers $\theta$, $\psi$, and $\omega$,
			%\State  \textbf{Repeat}
		\Repeat
	%	\State/*The Offloading Decision*/
		\State Based on $\mathbf{f}^{k-1}$, each vehicle obtains its  selection decision by solving (\ref{ua2}) and using rounding method;
		% \State /* Optimization of  Computation Resource and Offloading Ratio*/
		 % \State  		  /*The Computation Resource*/
		 % \State Initialize the Lagrangian multipliers $\theta$, $\psi$, $\varrho$, and $\omega$;
		 %  \State Initialize the Lagrangian multipliers $\upsilon$, $\mu$, and $\eta$;
		%   \State Set $t = 0$ and $t' = 0$
		 \Repeat \label{5}
		% \For {each $i\in\mathcal{N}$}
		 \State Update $\theta(t+1)$, $\psi(t+1)$, $\varrho(t+1)$ and $\omega(t+1)$  based on (\ref{iterative});
		 \State Calculate $\mathbf{f}^{(k)}$ using  (\ref{f_op});
		 		   		%	\For { each $j\in\mathcal{M}$}
		 		   		%	\If {$x_{ij}=1$}
		 		   		%		\State Update $\psi_j(t+1)$, $\varrho_{ij}(t+1)$ and $\omega_{ij}(t+1)$ based on (\ref{iterative});
		 		   		%		\State Calculate $f_{ij}^{(k)}$ based on (\ref{f_op});
		 		   	%		\EndIf
		 		   	%	\EndFor
		 		   %		\EndFor
		% \State Update $\theta_i(t+1)$, $\psi_j(t+1)$, $\varrho_{ij}(t+1)$ and $\omega_{ij}(t+1)$ based on (\ref{iterative});
		 %\State Calculate $\mathbf{f}^{(k)}$ based on (\ref{f_op});
		 \Until {Convergence.} \label{8}
		  % \State  		  /*The Offloading Ratio*/
		 % \Repeat
		   		% \State Update $\mu_{ij}(t+1)$ and $\eta_{ij}(t+1)$ based on (\ref{iterative2});
		   	%	\For {each $i\in\mathcal{N}$,each $j\in\mathcal{M}$}
		   		%	\For { each $j\in\mathcal{M}$}
		   		%	\State Each vehicle calculates ${\lambda}^{(k)}$ based on (\ref{lam_op});
		   			%	   		\EndFor
		   	%	\EndFor
		   		 \State k = k+1;
		   		 \Until {Convergence.}
%	\State  \textbf{Until}:{Convergence}
	\end{algorithmic}
\end{algorithm}

%Using JSCO, the (\ref{S2}) can be solved in a distributed manner where each vehicle $i$ solves its own subproblems and only exchanges information with its neighboring vehicles which are served by the same VEC server. Specifically, at each iteration, vehicle $i$ compute $x_{ij},j\in \mathcal{B}_i$, using the value $t_{ij}$. To compute the required computation resource $f_{ij}$, each vehicle $i$  needs the values of $x_{ij}$, $\lambda_{ij}$, and the neighboring  vehicles' computation resource at the previous iteration based on (\ref{iterative}), (\ref{f_op}). The offloading ratio $\lambda_{ij}$ can be computed for a given$f_{ij}$ and $x_{ij}$.

%Based on above analysis, we present our joint algorithm for  selection decision, computation resource, and offloading ratio (JSCO), which is summarized in Algorithm \ref{outer}. We   decouple (\ref{S2})  into  selection decision and optimization of  offloading ratio and computation resource.  First, based on  given  ${f}_{ij}$  and  ${\lambda}_{ij}$,  each vehicle determine ${x}_{ij}$ using relaxation and rounding method. Then,  under the obtained $x_{ij}$ and the neighboring vehicles's  computation resource at the previous iteration, each vehicle Lagrangian method to obtain $f_{ij}$  and  ${\lambda}_{ij}$ and repeat this process until convergence. 

	\section{Numerical Results}
\label{sr}
In this section, we present extensive simulation results to evaluate the performance of the proposed algorithm.
\subsection{Simulation Setup}
We consider a unidirectional road,  where $5$ RSUs are randomly located along a $100$-meter road. There are 40 arriving vehicles on the road, and they are running at a constant speed of $120$ km/hr. The bandwidth of each RSU is $1.25$ KHz.  The transmission power of each vehicle is 100 $mW$ and the noise power is $10^{-10} mW$.  We set the channel gain $h_{ij} = d_{ij}^{-4}$, where $d_{ij}$ is the distance between vehicle $i$ and RSU $j$.

Each RSU is equipped with a VEC server.  The computation resources of the VEC servers from the beginning of the road to the other end are $\{5,10,15,20,25\}$ GHz. Each vehicle has a computation task. The input data size,  required computation resources, and  maximum latency constraint of each computation task are uniformly distributed in the range of $U[100,300]$ KB, $U[0.5, 1.5]$ GHz, and  $U[8,10]$ s, respectively. The computation capability of each vehicle is $1$ GHz.

%The computation resources of the VEC servers from the beginning of the road to the other end are $\{\}$ 
To verify the performance of our proposed JOSC algorithm, we introduce the following benchmark schemes,
\begin{itemize}
%\item JOP: The scheme that relaxs the binary constraint of problem (\ref{S2}) which is a upper bound.
\item GS: Given a feasible computation resource, the Greedy Selection optimization scheme (GS)  greedily picks out the best VEC server  with maximal utility for each vehicle.
\item RA: The Resource Allocation scheme (RA)  optimizes computation resource, as in \cite{yu2016joint}. In this scheme, all vehicles offload their tasks to the nearest VEC server.
%\item BFS: The Brute Force Scheme  (BFS) explores all cases of selection decisions and feasible computation resource, and then chooses the VEC server with the minimum task completion time for each vehicle, which finds the global optimum.
\end{itemize}

\subsection{Performance Analysis}

\begin{figure}
	\centering
		\subfigure[Computation resource]{
		 	 	 		\includegraphics[width=3.2 in]{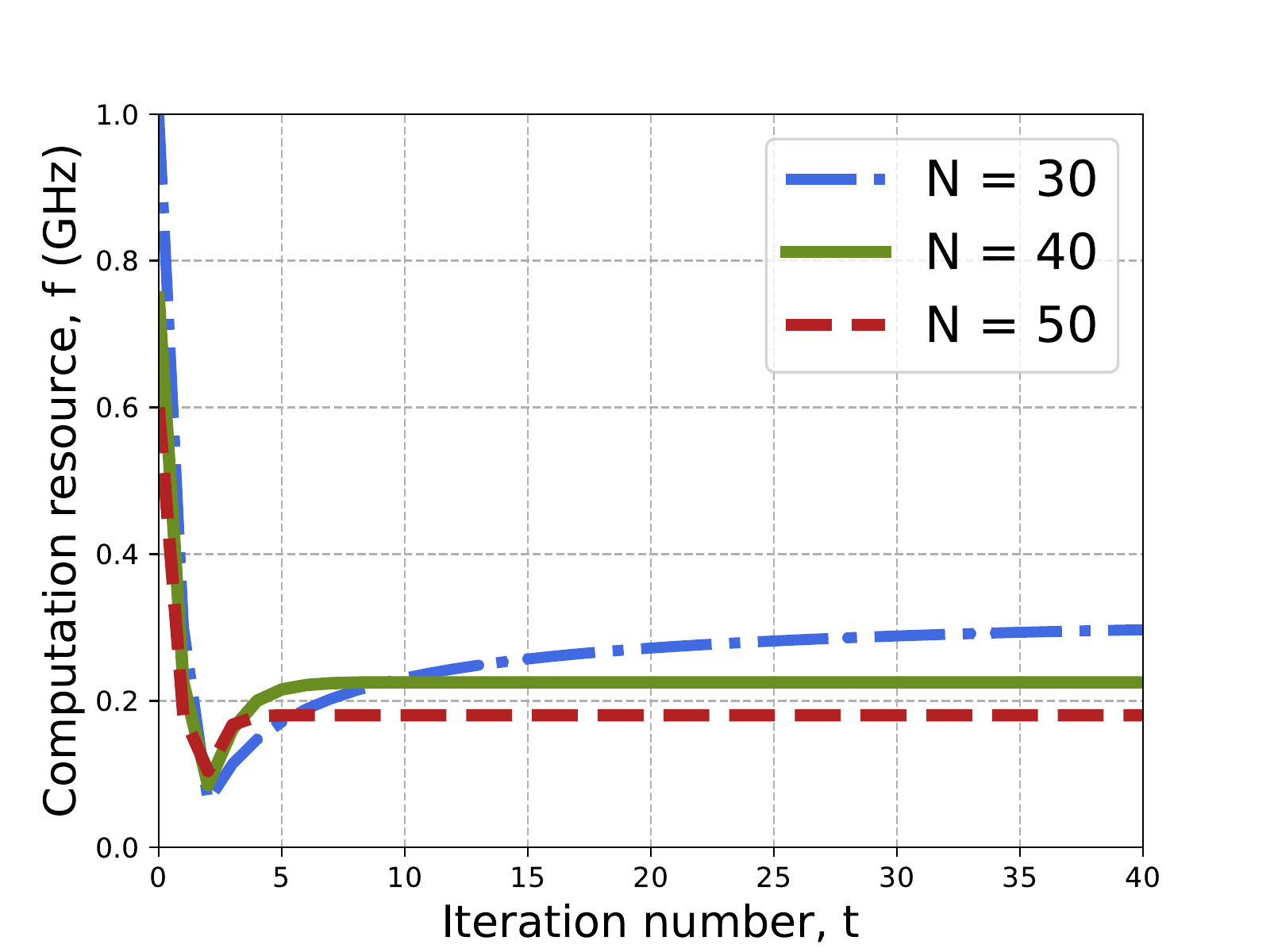}
		 	 	 		\label{fig2}
		 	 	 	}	 	 	 	
	\hspace{-0.325in}	 		
	\subfigure[System utility]{
	 		\includegraphics[width=3.2 in]{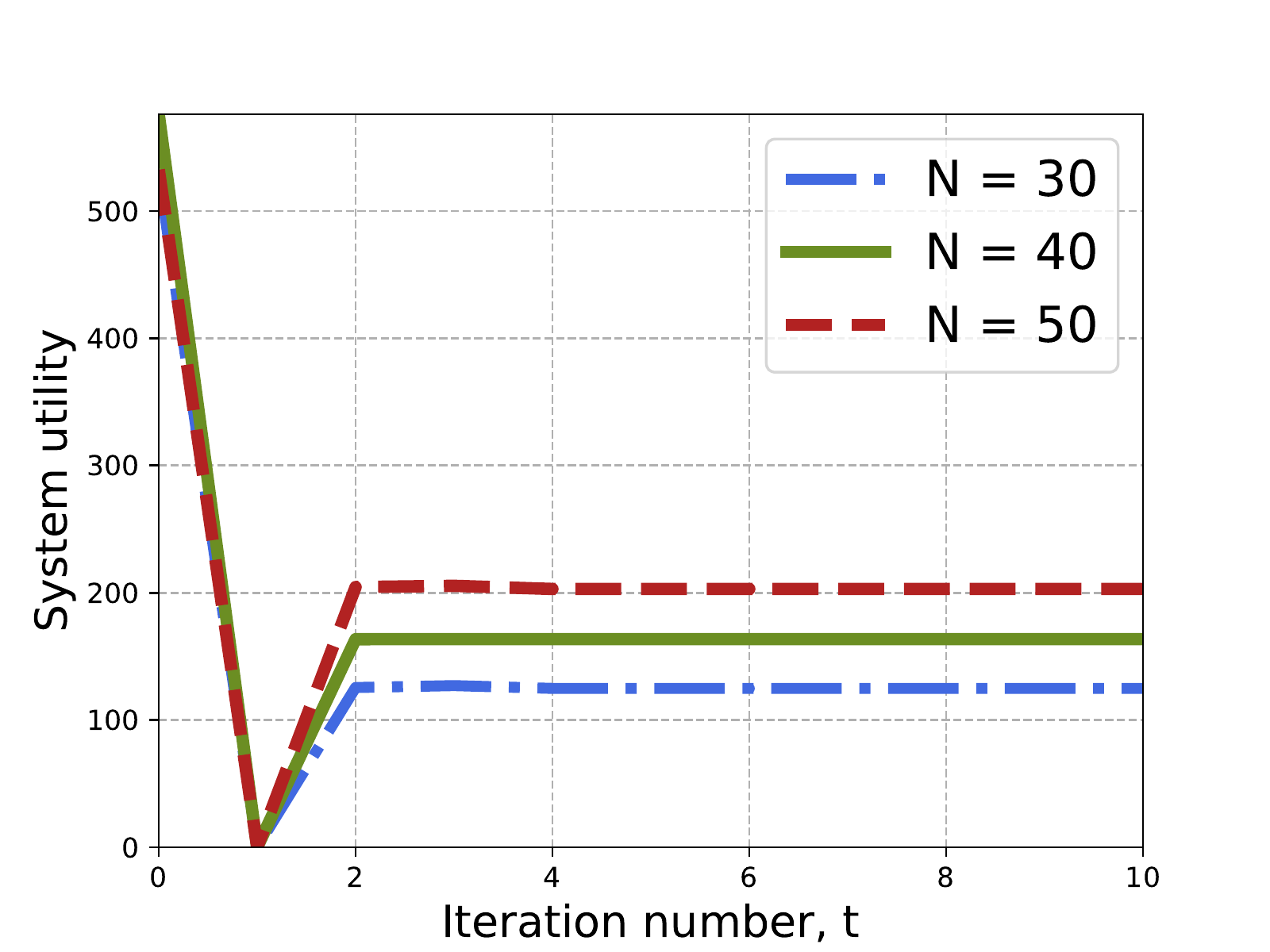}
	 		\label{fig5}
	 	}
	 	 	\caption{Convergence of JOSC under different  $N$.}

	\label{sfig}
\end{figure}

%\subsection{Convergence}
Fig. \ref{sfig} shows the convergence of our proposed JOSC algorithm. Specifically, Fig. \ref{fig2} plots the convergence of the inner loop of Algorithm \ref{outer}, i.e. the loop from step \ref{5} to step \ref{8} in JOSC.  We observe that the computation resource  can achieve converge very fast and a smaller value of  $N$  results in more computation resource.  From Fig. \ref{fig5}, we can see that the convergence of  system utility  can achieve converge within 4 iterations and  a larger value of  $N$  leads to  higher system utility. 
%\vspace{-0.1in}\

  Fig. \ref{fig3} shows system utility with respect to the number of arriving vehicles at the starting point  under different schemes. We can draw several observations from   Fig. \ref{fig3}. First, the system utility of the proposed JOSC  is obviously higher than GS and RA since it  jointly optimizes VEC server selection  and computation resource allocation. Second,  the system utility of  JOSC and GS increases rapidly  as the number of arriving vehicles increases while the system utility of  RA starts decreasing when the number of arriving vehicles is larger than $45$. The reason is that RA does not consider VEC server selection  which prolongs task processing time and decreases system utility. On the contrary, JOSC and GS improves system utility by performing selection decision. Moreover,  the system utility gap between JOSC and GS increases with increasing the number of vehicles. This is justified since the JOSC adopts the optimal policies of selection decision and computation resource but  GS only considers selection decision.

 Fig. \ref{fig4} depicts  load balancing performance   of JOSC, RA and GS for  the same profile (i.e. $N=40,M=4$). We can see that  the performance of JOSC is significantly better than that of GS and RA and RA does not completely balance the load among servers. While GS greedily chooses the VEC server which provides maximal utility, RA offloads  all tasks to the nearest VEC server (i.e. Server 1) which leads to an unbalanced assignment of resources.  On the contrary,  JOSC balances the number of vehicles served by each server and keeps a higher system utility. % In addition, it is reasonable that the load on each server is not absolutely equal  with different maximal computation resource ($F_j$) and different distance between vehicle and VEC server. % (i.e. a large $F_j$ will lead to the increasing of $U_i$ while a large distance  ). 
  
   \begin{figure}
   	\centering
   	\includegraphics[width =3.2 in]{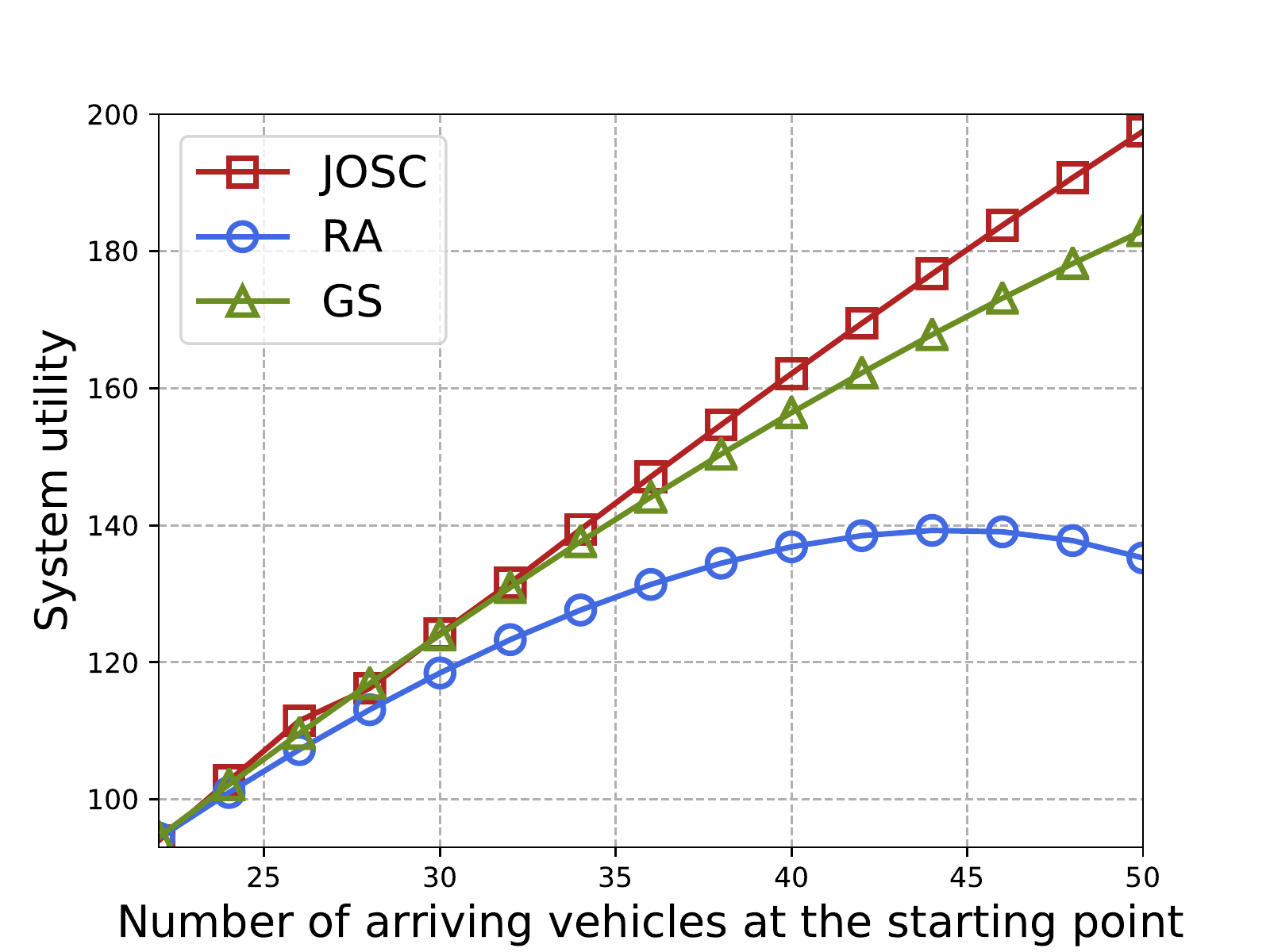}
   	\caption{Comparison of system utility of the number of arriving vehicles at the starting point under different schemes.}
   	\label{fig3}
   \end{figure}
 \begin{figure}
 	\centering
 	\includegraphics[width =3.2 in]{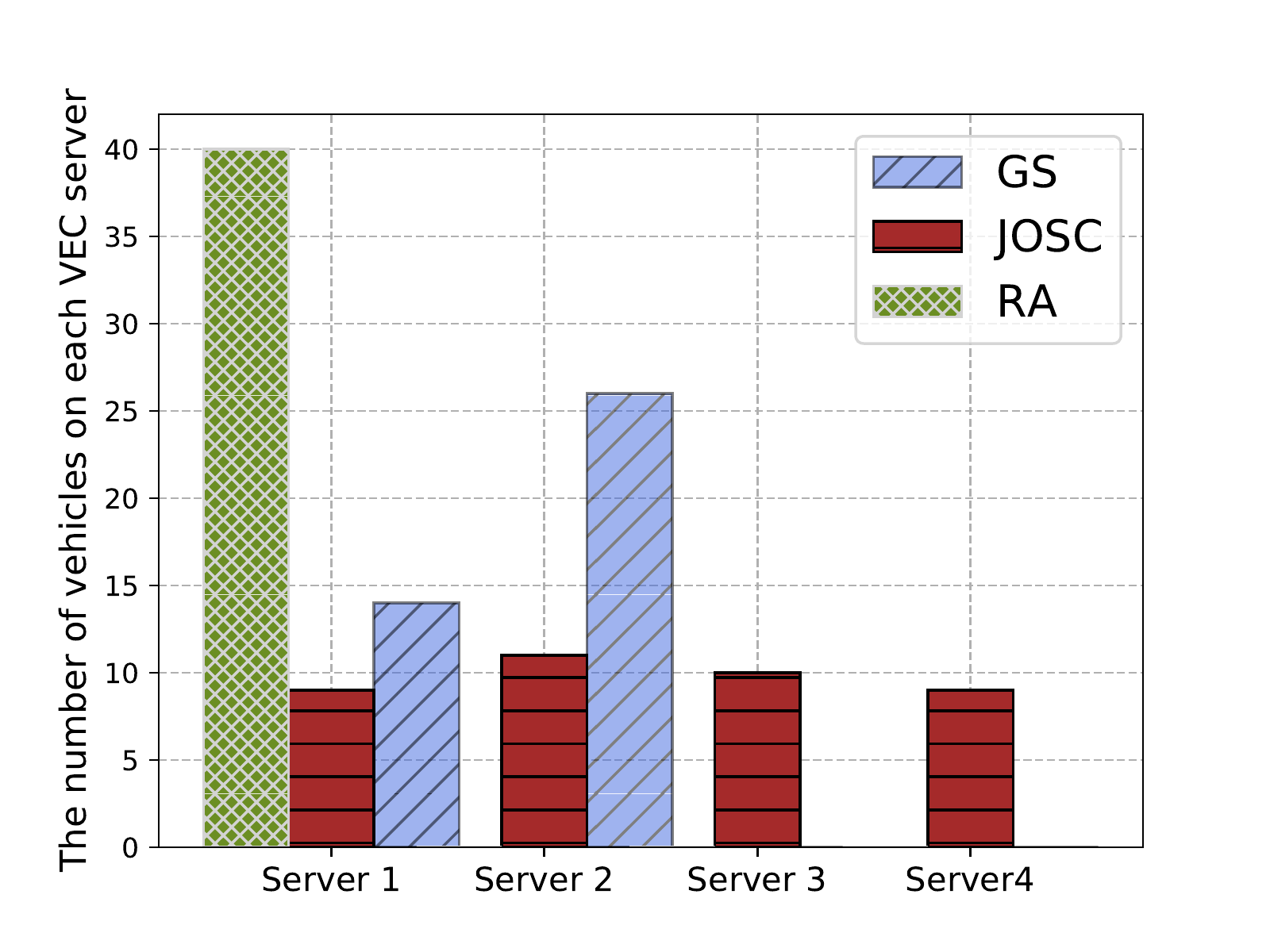}
 	\caption{Comparison of load balance  under different schemes with $N=40$ and $M=4$.}
 	\label{fig4}
 \end{figure}

	\section{Conclusions}
\label{c}
In this paper, we proposed a joint  offloading and resource allocation  approach for maximizing system utility in vehicular edge computing networks. We incorporated the components due to both communication and computation in the definition of the task processing delay.  We provided a low complexity, JOSC, by jointly optimizing selection offloading decision  and computation resource. Numerical results demonstrated that the proposed JOSC not only significantly outperforms the benchmark policies  in terms of system utility, but also performs better on balancing the load distribution, compared to the other counterparts.  %For the future work, we are going to consider a more general case that vehicles may arrive and depart dynamically at a unidirectional road. In this case, the  mobility patterns might be generalized  as a stochastic model in the problem formulation.
	
\begin{appendices}
\section{Proof of Lemma 1}	\label{appA}

According to constraints  (\ref{c2}) and (\ref{c7}), the objective function  ${\sum_{i=1}^{N}{\alpha \log(1+\beta-\sum_{j=0}^{M}x_{ij}T_{ij})}} $ is equivalent to  ${\sum_{i=1}^{N}\sum_{j=0}^{M}{\alpha \log(1+\beta-x_{ij}T_{ij})}}$.

Substituting (\ref{RSU})-(\ref{delay}) into (\ref{c1}), with $x_{i0}=1-\sum_{j=1}^{M}x_{ij}$, we have
\begin{equation}
\begin{split}
&\sum_{j=1}^{M}x_{ij}(T_{ij}-T_{i0})\leqslant T_{i}^{max}-T_{i0}\\
&\sum_{j=1}^{M}x_{ij}(\sum_{k=1}^{j-1}\frac{L_k}{v}+T_{ij}^{com}+\frac{c_i}{f_{ij}}-\frac{c_i}{f_i})\leqslant T_{i}^{max}-T_{i0}\\
&\sum_{j=1}^{M}x_{ij}(\varLambda_{ij}+\frac{c_i}{f_{ij}})\leqslant T_{i}'
\end{split}
\end{equation}
where $\varLambda_{ij} = \sum_{k=1}^{j-1}\dfrac{L_k}{v}+T_{ij}^{com}-\frac{c_i}{f_i}$ and  $T_i' = T_{i}^{max}-T_{i0}$. Thus, we obtain  constraint (\ref{c8}).

\section{Proof of Lemma 2}
\label{appB}
The second-order derivative of $U(\mathbf{f})$ with respect to $f_{ij}$ is
\begin{equation}
\begin{split}
\dfrac{\partial^2 U(\mathbf{f})}{\partial f_{ij}^2} = &-\,{\frac {2\alpha\,x_{{{\it ij}}}c_{{i}}}{f_{ij}^{3}\ln  \left( 2
 \right) } \left( 1+\beta-\frac{x_{ij}c_i}{f_{ij}}-x_{{{\it 
ij}}}\Lambda_{{{\it ij}}} \right) ^{-1}}\\&-{\frac {\alpha\,{x_{{{\it ij}
}}}^{2}{c_{{i}}}^{2}}{f_{ij}^{4}\ln  \left( 2 \right) } \left( 1+\beta-
\frac{x_{ij}c_i}{f_{ij}}-x_{{{\it ij}}}\Lambda_{{{\it ij}}}
 \right) ^{-2}}
\end{split}
\end{equation}
As $\alpha>0$, $x_{ij}c_i\geqslant 0$, and $\left( 1+\beta-\dfrac{x_{ij}c_i}{f_{ij}}-x_{{{\it 
ij}}}\Lambda_{{{\it ij}}} \right) > 0$, we have $\dfrac{\partial^2 U(\mathbf{f})}{\partial f_{ij}^2}\leqslant 0 $. Thus the objective function   $U(\mathbf{f})$ is  concave. Combining with the linear convex constraints,  (\ref{computation}) is a convex optimization problem.
\end{appendices}

	\bibliographystyle{IEEEtran}
	\bibliography{reference}
\end{document}